\documentclass[final]{dmtcs-episciences}

\usepackage[utf8]{inputenc}
\usepackage{subfigure}

\usepackage{amsmath,amsfonts,amssymb,amsthm,amscd}
\usepackage{enumerate}

\usepackage{hyperref}
\usepackage{breakurl}

\newtheorem{theorem}{Theorem}                             
            
\newtheorem{lemma}[theorem]{Lemma}
\newtheorem{corollary}[theorem]{Corollary}

\newtheorem{observation}[theorem]{Observation}

\theoremstyle{definition}

\def\myincludegraphics#1{\includegraphics{#1}}

\def\F{{\cal F}}

\newcommand{\btt}{\raisebox{-0.4ex}{\includegraphics{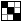}}}
\newcommand{\bttx}{\raisebox{-0.4ex}{\includegraphics{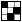}}}
\newcommand{\bdc}{\raisebox{-0.4ex}{\includegraphics{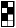}}}
\newcommand{\btd}{\raisebox{-0.4ex}{\includegraphics{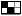}}}
\newcommand{\bddw}{\raisebox{-0.4ex}{\includegraphics{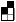}}}
\newcommand{\bdt}{\raisebox{+0.1ex}{\includegraphics{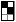}}}
\newcommand{\bccw}{\raisebox{-0.9ex}{\includegraphics{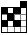}}}

\newcommand{\bdcw}{\raisebox{-0.5ex}{\includegraphics{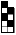}}}
\newcommand{\bct}{\raisebox{-0.4ex}{\includegraphics{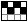}}}
\newcommand{\btc}{\raisebox{-0.4ex}{\includegraphics{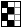}}}
\newcommand{\bdcd}{\raisebox{-0.4ex}{\includegraphics{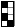}}}
\newcommand{\bjcw}{\raisebox{-0.4ex}{\includegraphics{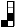}}}

\begin{document}

\title{Irreversible 2-conversion set in graphs of bounded degree\thanks{Most of the research was conducted during the DIMACS/DIMATIA REU program at Rutgers University in June and July 2007, supported by the project ME 886 of the Ministry of
Education of the Czech Republic.}
}

\author{
Jan Kyn\v{c}l\affiliationmark{1,2}\thanks{Partially supported by Swiss National Science Foundation Grants 200021-137574 and 200020-144531 and by the project CE-ITI (GA\v{C}R P202/12/G061) of the Czech Science Foundation.} 
\and
Bernard Lidick\'{y}\affiliationmark{1,3}\thanks{Partially supported by NSF grant DMS-1600390.}  
\and
Tom\'{a}\v{s} Vysko\v{c}il\affiliationmark{1,4}
} 


\affiliation{
  Charles University, Prague, Czech Republic\\
  EPFL, Switzerland\\
  Iowa State University, Ames IA, USA\\
  Rutgers University, Piscataway NJ, USA
}


\received{2016-11-29}
\revised{2017-6-8}
\accepted{2017-9-14}
\publicationdetails{19}{2017}{3}{5}{2559}

\maketitle


\begin{abstract}
An {\em irreversible $k$-threshold process\/} (also a {\em $k$-neighbor bootstrap percolation\/}) is a dynamic process
on a graph where vertices change color from white to black
if they have at least $k$ black neighbors.
An {\em irreversible $k$-conversion set\/} of a graph $G$ is a subset $S$ of vertices
of $G$ such that the irreversible $k$-threshold process starting with the vertices of $S$ black eventually
changes all vertices of $G$ to black.
We show that deciding the existence of an irreversible 2-conversion set
of a given size is NP-complete, even for graphs of maximum degree $4$,
which answers a question of Dreyer and Roberts.
Conversely, we show that for graphs of maximum degree $3$, the minimum size of an irreversible 2-conversion set can be computed in polynomial time.
Moreover, we find an optimal irreversible 3-conversion set for
the toroidal grid, simplifying constructions of Pike and Zou.
\end{abstract}


\section{Introduction}

Roberts~\cite{roberts-02} and Dreyer and Roberts~\cite{dreyer-09} recently studied mathematical modelling of the spread of infectious diseases. They
used the following model.

Let $G=(V,E)$ be a graph with vertices colored white and black.
An \emph{irreversible $k$-threshold process} is a process where vertices 
change color from white to black. More precisely, a white vertex becomes black
at time  $t+1$ if at least $k$ of its neighbors are black at time $t$.

An \emph{irreversible $k$-conversion set} $S$ is a subset of $V$ such that
the irreversible $k$-threshold process starting with the vertices of $S$ set to
black and all others white will result in a graph $G$ with all vertices black
after a finite number of steps. 

More general models of spread of infectious diseases and the complexity of the 
related problems were studied by Boros and Gurwich~\cite{boros-07}.

A natural question to ask is what is the minimum size of 
an irreversible $k$-conversion set in a graph $G$. 

\begin{center}
  \fbox{
   \parbox{0.97\columnwidth}{
     \begin{tabular}{lll}
       \multicolumn{3}{l} {Problem I$k$CS$(G,s)$:} \\
       \textbf{Input:}  & \multicolumn{2}{l}{a graph $G$ and a positive integer $s$} \\
       \textbf{Output:} & YES & if there exists an irreversible $k$-conversion set of size $s$ in $G$ \\
                        & NO & otherwise \\
     \end{tabular}
    }
  }
\end{center}

Dreyer and Roberts~\cite{dreyer-09} proved that I$k$CS is NP-complete for every fixed $k \ge 3$ 
by an easy reduction from  the independent set problem.
For $k = 1$ the problem is polynomial since one black vertex per connected
component is necessary and sufficient.
Dreyer and Roberts~\cite{dreyer-09} asked what is the complexity of the I$k$CS problem if $k = 2$.
As the first result of this paper we resolve this open question.

\begin{theorem}\label{thm:np}
The problem I2CS is NP-complete even for graphs of maximum degree $4$.
\end{theorem}

A subset $W$ of vertices of a graph $G=(V,E)$ is a \emph{vertex feedback set} if $V \setminus W$ is acyclic.
For 3-regular graphs, the I2CS problem is equivalent to finding a vertex feedback set, 
which can be solved in polynomial time~\cite{ueno-88}. 
We extend this result to graphs of maximum degree $3$.

\begin{theorem}\label{theorem_max3}
The problem I2CS is polynomially solvable for graphs of maximum degree $3$.
\end{theorem}

Boros and Gurwich~\cite{boros-07} proved that
determining the minimum size of the conversion set in graphs of maximum degree $3$ is NP-complete if every vertex has its own threshold; that is, when a part of the input is a function $t:V\rightarrow\mathbb{N}$ assigning the threshold for conversion individually to each vertex rather than having some universal threshold number $k$.

Note that the problem I2CS$(G,s)$ is polynomial if the maximum degree of $G$
is at most $2$ as a path of length $l$ requires $\lceil \frac{l+1}{2} \rceil$ black vertices and a cycle of
length $l$ requires $\lceil \frac{l}{2} \rceil$ vertices.

We also give a construction of an optimal 
irreversible 3-conversion set for a toroidal grid 
$T(m,n)$, which is the Cartesian product of the cycles $C_m$ and $C_n$.
Flocchini et al.~\cite{flocchini-04} and Luccio~\cite{luccio-98} gave lower and upper bounds differing by a linear $O(m+n)$ term; see also~\cite{dreyer-09}. Pike and Zou~\cite{pike-06} gave an optimal construction. We present a simpler optimal construction.

\begin{theorem}\label{thm:torus}
Let $T$ be a toroidal grid of size $m \times n$, where $m, n \ge 3$.
If $n=4$ or $m=4$ then $T$ has an irreversible 3-conversion set
of size at most $\frac{3mn + 4}{8}$.
Otherwise, $T$ has an irreversible 3-conversion set
of size at most $\frac{mn + 4}{3}$.
\end{theorem}

Theorems~\ref{thm:np} and~\ref{thm:torus} appeared in our early preprint~\cite{KLV09_nakaza}.
When preparing the final version of this paper, we found that Centeno et al.~\cite{CDPRS11_opisovaci} have published a different proof that the problem I2CS is NP-complete. The graph in their construction has maximum degree $11$.
Later an anonymous referee pointed to us that the NP-completeness of the I2CS problem was first proved by Chen~\cite{Ch09_approximability}, and that an independent proof of Theorem~\ref{thm:np} has been published by Penso et al.~\cite{PPRS_P3_convexity}.
Most recently, Takaoka and Ueno~\cite{TU15_irreversible} gave an alternative proof of Theorem~\ref{theorem_max3}, solving our problem stated in an earlier version of this paper.

Balogh and Pete~\cite{BaPe98_random_disease} reported tight asymptotic bounds on the minimum size of an irreversible $k$-conversion set in the $d$-dimensional integer grid. Balister, Bollob\'as, Johnson and Walters~\cite{BBJW10_majority} obtained more precise bounds for the case $d+1\le k\le 2d$.

An irreversible $k$-conversion process is equivalently also called a {\em ($k$-neighbor) bootstrap percolation}. Bootstrap percolation was introduced by Chalupa, Leath and Reich~\cite{CLR79_bootstrap_Bethe} as a model for interactions in magnetic materials. Bootstrap percolation theory is typically concerned with $d$-dimensional lattices (and in recent years, other classes of graphs as well) where each vertex is ``infected'' independently at random with some fixed probability.
See~\cite{Adler91_bootstrap} for an early review of the subject or~\cite{DeGLawDaw09_encyclopedia} for a recent survey.
See~\cite{BBDM12_sharp_threshold,Uzzell12_improved} for the most recent results for $d$-dimensional integer grids.

Several authors~\cite{Benevides-13, Marcilon-14} studied the computational complexity of the minimum number of steps of the bootstrap percolation needed to percolate the whole graph.

Many other variants of bootstrap percolation have been studied in the literature. Examples include hypergraph bootstrap percolation~\cite{BBMR12_lingebra} or weak $H$-saturation of graphs~\cite{Al85_problem,Bo68_weakly}.


\section{Irreversible 2-conversion set is NP-complete}

In this section we give a proof of Theorem~\ref{thm:np}.

The problem is in NP, since the verification whether $S \subseteq V$ is 
an irreversible 2-conversion set can be done by iterating the conversion process in at most $|V|$ steps. Hence the verification
can be done in a polynomial time.

In the rest of the proof we show that I2CS$(G,s)$ is 
NP-hard by a polynomial-time reduction from 3-SAT.
We introduce a variable gadget, a clause gadget and a gadget that
checks if all clause gadgets are satisfied. 

Since a white vertex needs two black neighbors to become black,
we have the following observation.

\begin{observation}\label{obs:degone}
Every irreversible 2-conversion set contains all vertices of degree $1$.
\end{observation}

According to this observation, in the figures of the 
gadgets we draw vertices of degree one black.

Let $\F$ be an instance of 3-SAT. We denote the number of variables by $n$
and the number of clauses by $m$.
We construct a graph $G_{\F}$ and give a number $s$ such that $\F$ is satisfiable if
and only if $G_{\F}$ has an irreversible 2-conversion set of size $s$.

\begin{figure}
\begin{center}
 \myincludegraphics{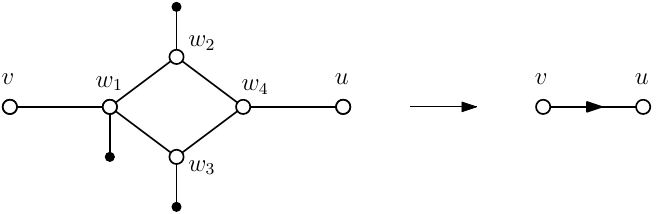}
\end{center}
\caption{A one-way gadget.\label{fig:oneway}}
\end{figure}

First we introduce a \emph{one-way} gadget; see Figure~\ref{fig:oneway}.
The gadget contains two vertices $u$ and $v$ which are called \emph{start} and \emph{end} of the one-way gadget.
Vertices  $w_1, w_2, w_3$ and $w_4$ are called \emph{internal} vertices of the gadget.

\begin{observation}
Let $u$ and $v$ be start and end of the one-way gadget. If internal vertices
are white at the beginning then the following holds:
\begin{enumerate}
\item If $v$ is black then $u$ gets a black neighbor from the gadget in three steps.
\item The vertex $w_4$ can become black only after $v$ becomes black.
\end{enumerate}
\end{observation}

We refer to the one-way gadget by a directed edge in
the following figures. Later, we set $s$ so that $S$ cannot contain
any internal vertices of one-ways. 
Thus, in the rest of the proof we assume that all internal vertices
are white at the beginning.

\begin{figure}
\begin{center}
 \myincludegraphics{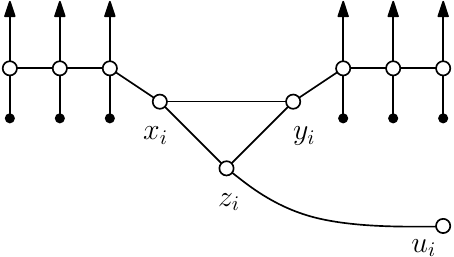}
\end{center}
\caption{A variable gadget $g(X_i)$ connected  to a vertex $u_i$ of a distributing path.\label{fig:variable}}
\end{figure}

\subsection{Variable gadget}
A gadget $g(X_i)$ for a variable $X_i$, where $1 \le i \le n$, consists of a triangle $x_iy_iz_i$ 
and two \emph{antennas}; see Figure~\ref{fig:variable}. 
The length of the antenna connected to $x_i$ and $y_i$ is equal
to the total number of occurrences of $X_i$ and $\neg X_i$, respectively,
taken over all clauses of $\F$. 
We call the white vertices of the $x_i$ antenna \emph{positive outputs}
and the white vertices of the $y_i$ antenna \emph{negative outputs}.
One-way gadgets with starts in the outputs have ends in clause gadgets.
The vertex $z_i$ is adjacent to a vertex $u_i$ lying on a distributing path,
which we define later.

We show that exactly one of $x_i$ and $y_i$ is black at the beginning.
This represents the value of the variable $X_i$.
The vertex $x_i$ corresponds to the true and $y_i$ to the false evaluation of $X_i$.
The purpose of the connection between $u_i$ and $z_i$ is to convert 
all vertices of the gadget to black if $\F$ is satisfiable.

\begin{observation}\label{obs:variable}
Let $S$ be an irreversible 2-conversion set. The gadget $g(X_i)$ has the following properties.
\begin{enumerate}
\item[(a)] If $x_i$ is black then all positive outputs will become black in the process. 
      Similarly for $y_i$ and negative outputs.
\item[(b)] If two of $x_i,y_i,z_i$ are black then all vertices of the gadget will become black
      in the process.
\item[(c)] $S$ must contain at least one of the vertices $x_i,y_i,z_i$. 
\item[(d)] If $S$ contains exactly one vertex of the gadget (except the vertices of degree $1$) 
      then it must be $x_i$ or $y_i$. 

\item[(e)] If $S$ contains exactly one vertex of the gadget then
      $z_i$ gets black only if $u_i$ gets black.

\end{enumerate}
\end{observation}

\begin{proof}
The first two properties are easy to check and hence we skip them.

First we check the property (c). Every vertex of the triangle $x_iy_iz_i$ has only one 
neighbor outside the triangle. Hence if all three vertices are white, they remain
white forever since each of them has at most one black neighbor.
Therefore $S$ must contain at least one of them.

Now we check the property (d). If $S$ contains only one of $\{x_i,y_i,z_i\}$
then all positive
and negative outputs are white at the beginning. Moreover, the positive 
outputs may become black only if $x_i$ gets black. Similarly for
negative outputs and $y_i$. 

Suppose for contradiction that $z_i \in S$. Then both $x_i$ and $y_i$ have only one
black neighbor ($z_i$) at the beginning. During the process the other
black neighbor has to be some output vertex which is not possible.
Hence $S$ must contain $x_i$ or $y_i$.

Finally, we check the property (e). By (d) we know that $z_i$ is white
at the beginning. Assume without loss of generality that $y_i$ is also
white while $x_i$ is black. The vertex $z_i$ can get black if $y_i$ or $u_i$
gets black. So assume for contradiction that $y_i$ gets black before
$z_i$. The only possibility is that the vertex from the antenna adjacent 
to $y_i$ gets black. 
But it is not possible since output vertices are white at the
beginning and they are connected to the rest of the graph by one-ways.
\end{proof}

Note that if $x_i$ or 
$y_i$ is in $S$ then it is still possible that the process
converts all vertices of the gadget to black, as the vertex $u_i$ may become black 
during the process.

Let $L$ be the set of all degree $1$ vertices in $G_{\F}$.
We set the parameter $s$ to $|L| + n$. Thus every variable
gadget has exactly one of $x_i$ and $y_i$ black at the beginning
and all other vertices of $G_\F$ of degree at least 2 are white.
We compute $|L|$ after we describe all the remaining gadgets.


\begin{figure}
\begin{center}
 \myincludegraphics{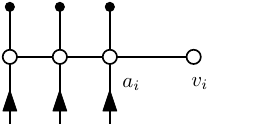}
\end{center}
\caption{A clause gadget $g(C_i)$ connected to a vertex $v_i$ of a collecting path. \label{fig:clause}}
\end{figure}

\subsection{Clause gadget}
The gadget $g(C_i)$ for a clause $C_i = (L_o \vee L_p \vee L_q)$,
where $1 \le i \le m$ and $L_o, L_p,$ and $L_q$ are literals, is depicted in Figure~\ref{fig:clause}.
The gadget consists of a path on three vertices corresponding to the three literals in the clause.
We call the path the \emph{spine} of the clause gadget.
Each vertex of the spine has one neighbor of degree $1$ and is connected to
the gadget of the corresponding variable by a one-way. 
The vertex of a clause corresponding to a literal $X_i$ 
is connected to a positive output of $g(X_i)$
and the vertex corresponding to a literal $\neg X_i$ is connected to a negative output of $g(X_i)$.
Finally, one vertex of the spine denoted by $a_i$ is connected to a vertex $v_i$
of a collecting path, which is defined later.

\begin{observation}
If one vertex of the spine is black, then all vertices of
the clause gadget get black in the process.
\end{observation}


\begin{figure}
\begin{center}
 \myincludegraphics{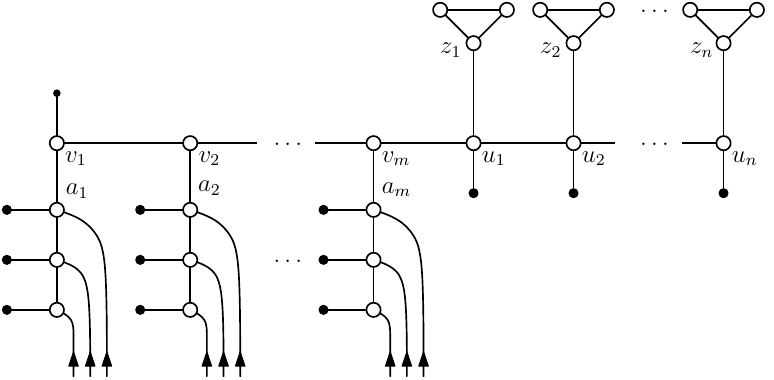} 
\end{center}
\caption{A collecting path $v_1, \ldots, v_m$ and a distributing path  $u_1, \ldots, u_n$ \label{fig:connection}}
\end{figure}

\subsection{Collecting and distributing gadget}

A \emph{collecting path} is a path on $m$ vertices $v_1, \ldots, v_m$ where each $v_i$ is connected to a clause gadget. Moreover, the vertex $v_1$ is also 
connected to a vertex of degree $1$. 
A \emph{distributing path} is a path on  $n$ vertices $u_{1}, \ldots, u_{n}$.
Each $u_i$ is connected to a vertex of degree $1$ and to the vertex $z_i$
of the variable gadget $g(X_i)$.
Finally, $v_m$ is connected to $u_1$; see Figure~\ref{fig:connection}. See Figure~\ref{fig:example} for an example of the whole graph $G_\F$.

\begin{observation}\label{obs:paths}
If the vertices of the distributing and collecting paths are white at the beginning, they will become all black in the process
if and only if all the clause gadgets get black during the process.
Moreover, for each vertex $v_i$ of the collecting path, it holds that $v_i$ cannot get black before its corresponding clause gadget becomes black.
\end{observation}

\begin{proof}
If all spines of clause gadgets are black then it is easy to observe that
the vertices of the collecting path get black in at most $m$ steps
from $v_1$ to $v_m$.
Once $v_m$ is black all the vertices
of the distributing path get black in at most $n$ steps
from $u_1$ to $u_n$.
It remains to check that $v_i$ cannot get black before 
a neighboring vertex $a_i$ gets black.

We start by checking the vertices of the distributing path. 
By Observation~\ref{obs:variable}(e), the vertex $z_n$ cannot get black before $u_n$.
Thus $u_n$ cannot get black before $u_{n-1}$ because $u_{n-1}$ is one of 
the two remaining neighbors which can be black before $u_n$.
Similarly, for $0 < i < n$, the vertices $z_i$ and $u_{i+1}$ cannot get black before
$u_{i}$. Also $u_1$ cannot get black before $v_m$.

Analogously, no vertex $v_i$, $2 \le i \le m$, of the collecting path 
can get black before $v_{i-1}$ and $a_i$ are both black.
For $i = 1$ we get that $a_1$ must get black before $v_1$.
\end{proof}

The graph $G_{\F}=(V,E)$ corresponding to the 3-SAT instance $\F$ constructed from these gadgets has a linear size in the size of $\cal F$. 
The size of $L$ is $15m+n+1$. 
Thus $s$ is set to $n+|L| = 15m+2n+1$.

\begin{figure}
\begin{center}
 \myincludegraphics{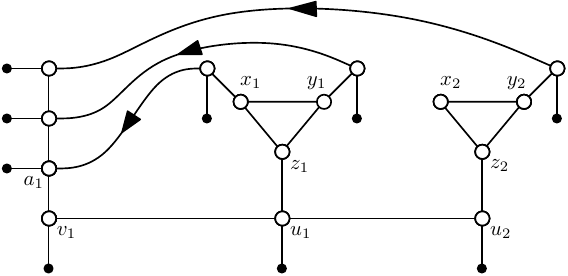} 
\end{center}
\caption{A graph $G_\F$ for the formula 
$\F = (X_1 \vee \neg X_1 \vee \neg X_2)$.
\label{fig:example}
}
\end{figure}

\begin{lemma}
If $\F$ is satisfiable then there exists an irreversible 2-conversion set $S$
of size $n + |L|$ in $G_{\F}$.
\end{lemma}
\begin{proof}
Let $\sigma: \{X_1, \dots, X_n\}\rightarrow \{\text{true}, \text{false}\}$ be a satisfying assignment of $\F$.
We choose the set $S$ to contain the set $L$ of leaves of $G_{\F}$ and one vertex from each variable gadget; from each gadget $g(X_i)$ we choose $x_i$ if $\sigma(X_i)=\text{true}$ and $y_i$ if $\sigma(X_i)=\text{false}$.
Since $\sigma$ is a satisfying assignment, after a finite number of steps every clause gadget
gets at least one black vertex in its spine. Then in at most two steps
all clause gadgets are completely black. Next the collecting path
gets black in at most $m$ steps and the distributing path gets black 
in another $n$ steps. 
Meanwhile, for each $i\in [n]$, the vertex $z_i$ gets two black neighbors and so it also 
gets black. The remaining white vertex of the pair $x_i, y_i$ gets
black in the next step.
Finally, also the remaining
antennas for every variable get black. 
Hence all vertices of $G_{\F}$ get black in the process.
\end{proof}

\begin{lemma}
If $\F$ is not satisfiable then there is no irreversible 2-conversion set of size $n + |L|$ in $G_{\F}$.
\end{lemma}
\begin{proof}
Assume for contradiction that there exists
an irreversible 2-conversion set $S$ of size $n + |L|$ in $G_{\F}$.  
By Observation~\ref{obs:degone}, we have $L \subseteq S$.
Due to Observation~\ref{obs:variable}(c), each variable gadget contains exactly one non-leaf vertex of $S$. So by Observation~\ref{obs:variable}(d), $S$ must contain exactly one of $\{x_i, y_i\}$ for each $i \in [n]$.
Hence there are no other black vertices.
We derive the truth assignment of the variables in the following way.
We set $X_i$ true if $x_i \in S$ and false if $y_i \in S$.

Let $C = (L_o \vee L_p \vee L_q) $ be a clause of $\F$. The gadget $g(C)$ gets black after
a finite number of steps of the process. 
By Observation~\ref{obs:paths}, $g(C)$ got black because of one of $g(X_o)$, $g(X_p)$
or $g(X_q)$. 
Hence $C$ is evaluated as true in $\F$. 
Therefore all clauses of $\F$ are evaluated as true which is a contradiction
with the assumption that $\F$ is not satisfiable.
\end{proof}

The proof of Theorem~\ref{thm:np} is now finished.


\section{Graphs with maximum degree $3$}\label{section3}

In this section we give a proof of Theorem~\ref{theorem_max3}. We follow the approach of Ueno, Kajitani and Gotoh~\cite{ueno-88}.

Let $G$ be a graph with maximum degree $3$. Without loss of generality, we assume that $G$ is connected, since a minimum irreversible $2$-conversion set can be computed for each component separately.
First we reduce the problem to graphs with minimum degree $2$. Let $H_5$ be the graph with five vertices and seven edges consisting of the cycle $v_1v_2v_3v_4v_5$ and the edges $v_2v_4$ and $v_3v_5$. 
Let $G_2$ be the graph obtained from $G$ by attaching a copy of $H_5$ to each vertex $v$ of $G$ of degree $1$ by identifying $v$ with $v_1$; see Figure~\ref{fig_attaching}. Observe that $G_2$ is a graph with maximum degree $3$ and minimum degree at least $2$.

\begin{figure}
\begin{center}
 \ifpdf\includegraphics[scale=1]{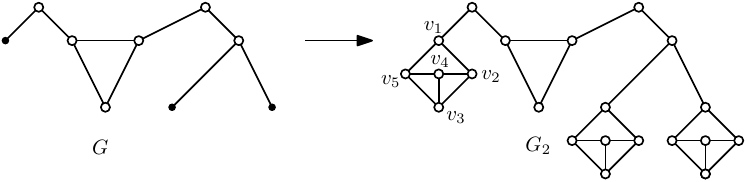}\fi
\end{center}
 \caption{A reduction of the I2CS problem to graphs with all degrees at least $2$.}
 \label{fig_attaching}
\end{figure}

For any graph $H$, let $C_2(H)$ be the minimum size of an irreversible $2$-conversion set in $H$.

\begin{lemma}\label{lemma_privesky}
Let $k$ be the number of vertices of degree $1$ in $G$. Then $C_2(G_2)=C_2(G)+k$.
\end{lemma}

\begin{proof}
By Observation~\ref{obs:degone}, every irreversible $2$-conversion set in $G$ contains all vertices of degree $1$. Since every irreversible $2$-conversion set in $G_2$ intersects all cycles, it contains at least two vertices in each copy of $H_5$. On the other hand, $\{v_3,v_4\}$ is an irreversible $2$-conversion set of $H_5$. It follows that by attaching a copy of $H_5$ to a vertex of degree $1$, the minimum size of an irreversible $2$-conversion set grows by $1$.
\end{proof}

If $G_2$ is $3$-regular, the I2CS problem is equivalent to finding a vertex feedback set, which can be solved in polynomial time by the result of Ueno, Kajitani and Gotoh~\cite{ueno-88}.

Now we take care of the case when $G_2$ has exactly one or two vertices of degree $2$.
First we consider the special case when the two vertices of degree $2$ are connected by an edge. We subdivide this edge by a new vertex $x$ and add one more vertex $y$ joined to $x$, forming a graph $G_1$. See Figure~\ref{fig_edge}.

\begin{figure}
\begin{center}
 \ifpdf\includegraphics[scale=1]{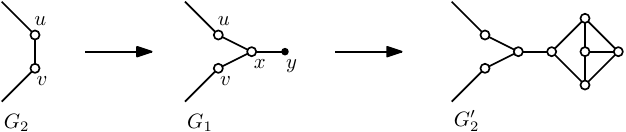}\fi
\end{center}
 \caption{The case of two adjacent vertices of degree $2$.}
 \label{fig_edge}
\end{figure}

\begin{lemma}\label{lemma_two2_joined}
Suppose that $u$ and $v$ are the only vertices of degree $2$ in $G_2$, and that $uv$ is an edge of $G_2$. Let $G_1$ be the graph $(V(G_2)\cup\{x,y\}, (E(G_2)\setminus\{uv\})\cup\{ux,vx,xy\})$. Then $C_2(G_1)=C_2(G_2)+1$.
\end{lemma}

\begin{proof}
If $S$ is an irreversible $2$-conversion set in $G_2$, then $S$ contains at least one of the vertices $u,v$. We claim that $S\cup\{y\}$ is an irreversible $2$-conversion set in $G_1$. This is clear if both $u$ and $v$ are in $S$. If exactly one of the vertices $u,v$ is in $S$, say, $u\in S$ and $v\notin S$, then $x$ turns black in the first step. Therefore, it is sufficient to show that $S\cup \{x,y\}$ is an irreversible $2$-conversion set in $G_1$. But this follows since in this case, the irreversible $2$-conversion process on the subset $V(G_2)$ is identical to the process on $G_2$ starting with $S$ black. 

Conversely, let $S'$ be an irreversible $2$-conversion set in $G_1$. Then necessarily $y \in S'$. We may assume that $x\notin S'$, otherwise we may replace $x$ by $u$ or $v$, or remove $x$ from $S'$ if both $u$ and $v$ are in $S'$. We claim that $S'\setminus\{y\}$ is an irreversible $2$-conversion set in $G_2$. Clearly, at least one of the vertices $u,v$, say, $u$, is in $S'$. If also $v\in S'$, the claim follows immediately. If $v\notin S'$, then during the irreversible $2$-conversion process on $G_1$, the vertex $v$ turns black only after its neighbor in $G_2$ other than $u$ turns black. Therefore, the irreversible $2$-conversion process on $G_2$ starting with $S'\setminus\{y\}$ black will be induced by the process on $G_1$ starting with $S'$ black.
\end{proof}

Modifying $G_1$ like in Lemma~\ref{lemma_privesky}, that is, by attaching a copy of $H_5$ to $y$, we obtain a graph $G'_2$ with two nonadjacent vertices of degree $2$ and all other vertices of degree $3$, satisfying $C_2(G'_2)=C_2(G_1)+1=C_2(G_2)+2$. 

Now we consider the case of two nonadjacent vertices of degree $2$.

\begin{lemma}\label{lemma_two2}
Suppose that $u$ and $v$ are the only vertices of degree $2$ in $G_2$ and that $u$ and $v$ are not adjacent. Then the graph $G_3$ obtained from $G_2$ by adding the edge $uv$ satisfies $C_2(G_3)=C_2(G_2)$. In particular, every irreversible $2$-conversion set in $G_3$ is also an irreversible $2$-conversion set in $G_2$. 
\end{lemma}

\begin{proof}
The inequality $C_2(G_2)\ge C_2(G_3)$ follows from the fact that $G_2$ is a subgraph of $G_3$. For the other inequality, suppose that $S$ is an irreversible $2$-conversion set in $G_3$. We show that then $S$ is also an irreversible $2$-conversion set in $G_2$. Every component of $G_3-S$ is a tree. In the beginning, the vertices of $S$ are black and the other vertices are white. In each step, the irreversible $2$-threshold process converts all isolated vertices and all leaves of the white subgraph of $G_3$ to black vertices.
If $w$ is an isolated vertex in $G_3-S$, $w$ has still at least two black neighbors in $G_2$, so it is converted to a black vertex in the first step.

Let $T$ be a tree component $T$ of $G_3-S$ with at least two vertices. If $uv$ is an edge of $T$, the two components of $T-uv$ will still be converted to black vertices in $G_2$, with $u$ and $v$ being the last vertices to be converted. If $u\in T$ and $v\notin T$, then all vertices of $T$ will still be converted to black vertices, with $u$ being the last vertex to be converted.
\end{proof}

We note that we could use Lemma~\ref{lemma_two2} also in the case when $uv$ is an edge, if we allowed multigraphs. However, we have decided not to use multigraphs in this paper.

The case of exactly one vertex of degree $2$ can be easily solved using Lemma~\ref{lemma_two2} by taking two disjoint copies of $G_2$.

\begin{corollary}
Suppose that $v$ is the only vertex of degree $2$ in $G_2$. Then the graph $G_3$ obtained from $G_2$ by adding a disjoint graph $G'_2$ isomorphic to $G_2$ and joining the vertex $v'$ of degree $2$ in $G'_2$ with $v$ by an edge is $3$-regular and satisfies $C_2(G_3)=2C_2(G_2)$.
\end{corollary}

We are left with the case when $G_2$ has at least $k\ge 3$ vertices of degree $2$. In this case, we construct a $3$-regular graph $G_3$ by attaching a caterpillar $T$ with $k$ leaves and $k-2$ vertices of degree $3$ forming the spine; see Figure~\ref{fig_caterpillar}. Every leaf of $T$ is identified with one vertex of degree $2$ in $G_2$. Let $V_2$ be the vertex set of $G_2$ and let $V_3$ be the vertex set of $G_3$. The graph $G_3-V_2$ is a path induced by the $k-2$ branching nodes of $T$.

\begin{figure}
\begin{center}
 \ifpdf\includegraphics[scale=1]{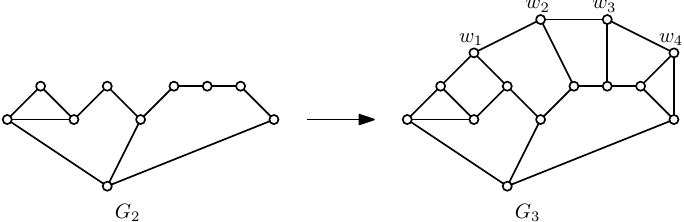}\fi
\end{center}
 \caption{Extending $G_2$ to a $3$-regular graph $G_3$ by attaching a caterpillar.}
 \label{fig_caterpillar}
\end{figure}

Let $\mu(G)$ be the \emph{cyclomatic number} of $G$. That is, $\mu(G)=e(G)-v(G)+\kappa(G)$, where $e(G)$, $v(G)$ and $\kappa(G)$ are the numbers of edges, vertices and components of $G$, respectively.

Define a function $f:2^{V_3} \rightarrow \mathbb{Z}$ from the set of subsets of vertices of $G_3$ as $f(X)\mathrel{\mathop:}=\mu(G_3)-\mu(G_3-X)$. Roughly speaking, $f$ measures the number of cycles broken by $X$ in $G_3$.
Let $f_2:2^{V_2} \rightarrow \mathbb{Z}$ be the restriction of $f$ to subsets of $V_2$.
Ueno, Kajitani and Gotoh~\cite{ueno-88} proved that $(V_3,f)$ is a linear $2$-polymatroid, using a linear representation of the dual matroid of the graphic matroid of $G_3$. More precisely, each vertex $v$ of $G_3$ can be represented as a $2$-dimensional subspace $h(v)$ of a certain vector space (over any field, and of dimension not exceeding the number of vertices of $G_3$) so that for each $X\subset 2^{V_3}$, the value $f(X)$ is equal to the dimension of the span of $\bigcup\{h(v); v\in X\}$. The function $f$ is called the \emph{rank function} of the $2$-polymatroid.
Since $f_2$ is a restriction of $f$, it follows that $(V_2,f_2)$ is also a linear $2$-polymatroid.

A set $M$ is a \emph{matching} in a $2$-polymatroid with rank function $f$ if $f(M)=2|M|$. A set $S$ is \emph{spanning} in a $2$-polymatroid $(V,f)$ if $f(S)=f(V)$. Let $\nu(V,f)$ be the maximum size of a matching in $(V,f)$ and let $\rho(V,f)$ be the minimum size of a spanning set of $(V,f)$.
Lov\'asz~\cite{lov81,LP86_matching_theory} proved the following generalization of Gallai's identity.

\begin{lemma}[{\cite{lov81}},{\cite[Lemma 11.1.1.]{LP86_matching_theory}}]
\label{lemma_gallai}
For every $2$-polymatroid $(V,f)$, we have $\nu(V,f)+\rho(V,f)=f(V)$.
\end{lemma}

Lov\'asz~\cite{lov81} proved that a maximum matching in a linear $2$-polymatroid can be found in polynomial time. It follows that also $\rho(V,f)$ can be computed in polynomial time, for any linear $2$-polymatroid $(V,f)$.

The theorem now follows from the following fact, generalizing~\cite[Theorem 3]{ueno-88}.

\begin{lemma}
\label{lemma_V2}
A set $S\subseteq V_2$ is a spanning set in $(V_2,f_2)$ if and only if it is an irreversible $2$-conversion set in $G_2$.
\end{lemma}

To prove the lemma, we use the following simple observation.

\begin{observation}
\label{obs_del_deg2}
Let $S$ be an irreversible $2$-conversion set in a graph $G$. Let $v$ be a vertex of $G$ of degree $2$ such that $v\notin S$. Then $S$ is an irreversible $2$-conversion set in $G-v$.
\qed
\end{observation}

\begin{proof}[Proof of Lemma~\ref{lemma_V2}.]
Let $S\subseteq V_2$ be an irreversible $2$-conversion set in $G_2$. We claim that $S$ is also an irreversible $2$-conversion set in $G_3$. If not, then $G_3-S$ contains a cycle $C$ of white vertices that will not be converted to black vertices during the irreversible $2$-threshold process starting with $S$ black. Since $G_3-V_2$ is a tree, $C$ contains a vertex of $V_2$; a contradiction. Therefore, $S$ is a feedback vertex set in $G_3$, equivalently, $G_3-S$ is acyclic, and this is equivalent to the fact that $f(S)=f(V_3)$. Since $G_3-V_2$ is acyclic, we have $f_2(S)=f(S)=f(V_3)=f(V_2)=f_2(V_2)$, and so $S$ is spanning in $(V_2,f_2)$.

Now let $S\subseteq V_2$ be a spanning set in $(V_2,f_2)$. By the previous arguments, this is equivalent to the fact that $S$ is an irreversible $2$-conversion set in $G_3$. 
Let $w_1w_2\dots w_{k-2}$ be the path $G_3-V_2$. Let $e$ be an edge joining $w_{k-2}$ with a vertex of $V_2$. By Lemma~\ref{lemma_two2}, $S$ is an irreversible $2$-conversion set in $G_3-e$. By Observation~\ref{obs_del_deg2}, $S$ is an irreversible $2$-conversion set in $(G_3-e)-w_{k-2}=G_3-w_{k-2}$. Similarly, by a repeated application of Observation~\ref{obs_del_deg2}, $S$ is an irreversible $2$-conversion set in  $G_3-w_{k-2}-w_{k-3}-\cdots -w_1 = G_2$.
\end{proof}

\subsection{Running time of the algorithm}

Lov\'asz and Plummer~\cite{LP86_matching_theory} estimated the running time of 
Lov\'asz's algorithm~\cite{lov81} to be $O(n^{17})$. Since every $2$-dimensional subspace can be represented by a pair of linearly independent vectors, the matching problem for $2$-polymatroids is equivalent to the \emph{linear matroid parity problem}, whose input is a linear matroid with a partition of its edges into pairs, and the goal is to find an independent set with maximum number of pairs. For matroids with $n$ elements and rank $r$, 
Gabow and Stallmann~\cite{GaSt85_linear_parity, GaSt86_linear_parity} gave an algorithm for the linear matroid parity problem running in time $O(nr^{\omega})$, where $O(n^{\omega})$ is the complexity of multiplication of two square $n\times n$ matrices. Moreover, for graphic matroids, Gabow and Stallmann~\cite{GaSt85_linear_parity, GaSt86_linear_parity} gave a very fast algorithm running in time $O(nr\log^6 r)$. They noted that the same algorithm can be used to solve the linear matroid parity problem for cographic matroids, that is, the duals of graphic matroids. This follows from the simple fact that a maximum matching $M$ and a basis $B$ containing $M$ determine a unique maximum matching in the dual matroid in the complement of $B$. 
Takaoka, Tayu and Ueno~\cite{TTU13_feedback} used Gabow's and Stallmann's algorithm to show that the vertex feedback set problem for graphs of maximum degree $3$ can be solved in time $O(n^2\log^6 n)$. However, we are not able to make a similar conclusion for the I2CS problem, and can guarantee only $O(n^{1+\omega})$ running time using the more general algorithm by Gabow and Stallmann.

Although the matroid $(V_3,f)$ constructed by Ueno, Kajitani and Gotoh~\cite{ueno-88} is cographic, our submatroid $(V_2,f_2)$ is not cographic in general. Moreover, $\rho(V_3,f)$ can be smaller than $\rho(V_2,f_2)$; see Figure~\ref{fig_counterexample}. Therefore, we cannot directly use the value $\nu(V_3,f)$ computed by the faster algorithm by Gabow and Stallmann.

\begin{figure}
\begin{center}
 \ifpdf\includegraphics[scale=1]{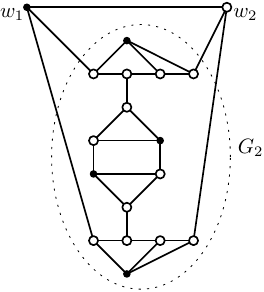}\fi
\end{center}
 \caption{A $3$-regular graph $G_3$ with an irreversible $2$-conversion set of size $5$. Every such set has to contain at least one of the vertices $w_1, w_2$; therefore, the connected subgraph $G_2 = G_3-w_1-w_2$ has no irreversible $2$-conversion set of size $5$.}
 \label{fig_counterexample}
\end{figure}

\subsection{Alternative proofs of Theorem~\ref{theorem_max3}}
Very recently, Takaoka and Ueno~\cite{TU15_irreversible} solved the I2CS problem on graphs of maximum degree $3$ by a reduction to the graphic matroid parity problem, which yields an algorithm running in time in time $O(n^2\log^6 n)$.

An anonymous reviewer suggested the following similar, more straightforward, reduction, using recent results of Kociumaka and Pilipczuk~\cite{KP14_faster} or Cao, Chen and Liu~\cite{CCL15_feedback}.

First, given an instance of the I2CS problem on a graph $G$ of maximum degree $3$, we reduce the problem to a graph $G_2$ with minimum degree $2$, using Lemma~\ref{lemma_privesky}. Then we add a new vertex $w$ attached to each vertex of degree $2$ in $G_2$ and call the new graph $G_3$. It is not hard to see that a set $S\subseteq V(G_2)$ is an irreversible $2$-conversion set in $G_2$ if and only if the complement of $S$ in $G_2$ induces a forest where each component has at most one vertex of degree $2$ in $G_2$~\cite[Lemma 2]{TU15_irreversible}. Therefore, $S$ is an irreversible $2$-conversion set in $G_2$ if and only $S$ is a vertex feedback set in $G_3$~\cite[Lemma 3]{TU15_irreversible}.

Next, we subdivide every edge $e=uv$ with $u,v\in V(G_3)\setminus\{w\}$ by a new vertex $w_{e}$ and call the resulting graph $G_4$. We partition the vertex set of $G_4$ into two sets $U,D$, where $D$ contains all the original vertices of $V(G_2)$, and $U$ consists of $w$ and all the vertices $w_e$. Both sets $U,D$ induce an independent set in $G_4$, and each vertex in $D$ has degree $3$ in $G_4$. Moreover, a set $S\subseteq D$ is a vertex feedback set in $G_4$ if and only if it is a vertex feedback set in $G_3$. The problem of determining a smallest size of a vertex feedback set in $G_4$ that is a subset of $D$ can be solved in time $O(n^2\log^6 n)$ by the result of Cao, Chen and Liu~\cite[Theorem 3.4]{CCL15_feedback} or Kociumaka and Pilipczuk~\cite[Theorem 6]{KP14_faster}, even in a more general case when $U$ and $D$ induce forests.


\section{Irreversible 3-conversion set in toroidal grids}

In this section we show a construction of an irreversible $3$-conversion set $S$
that proves Theorem~\ref{thm:torus}. We denote the toroidal grid 
of size $n \times m$ by $T(n, m)$.
When the dimensions of the grid are clear from the context or not
important, we simply write  $T$ instead of $T(m,n)$.
We assume that the entries of the grid are \emph{squares} and two of them
are neighboring if they share an edge.
First we discuss the general case where $m \neq 4$ and $n \neq 4$.

We define a coordinate system on $T$ such that the left bottom corner is $[0,0]$.
A \emph{pattern} is a small and usually rectangular piece of a grid
where squares are black and white.
\emph{Placing a pattern $P$ at position} $[i,j]$ in $T$ means that the left bottom
square of $P$ is at $[i,j]$ in $T$. If a vertex of $T$ has color defined by
several patterns then it is white only if it is white in all the patterns.
We describe a rectangle of a grid by the coordinates of the bottom left corners of its bottom left and top right
squares.
\emph{Tiling} a rectangle $R$ by a pattern $P$ means placing several non-overlapping 
copies of $P$ to $R$ so that every square of $R$ is covered.

Let $m = 3k+a$ and $n = 3l + b$, where $a,b \in \{0,2,4\}$. Without loss of generality we assume that $a\le b$. By $g$ we denote
the greatest common divisor of $k$ and $l$.

For $i\in \{0,\dots, g-2\}$ we place a pattern \bttx{} at $[0,3i]$. Next we tile the rest
of the rectangle $[0,0][3k-1, 3l-1]$ by a pattern \btt. 
The remaining part of the grid can be decomposed
into three rectangles of dimensions $3k \times b$, $a \times 3l$ and $a \times b$
(some of them may be empty). 

We distinguish several cases depending on $a$ and $b$. They
are depicted in Figure~\ref{fig:torc} and their description follows.

\begin{figure}
\begin{center}
 \myincludegraphics{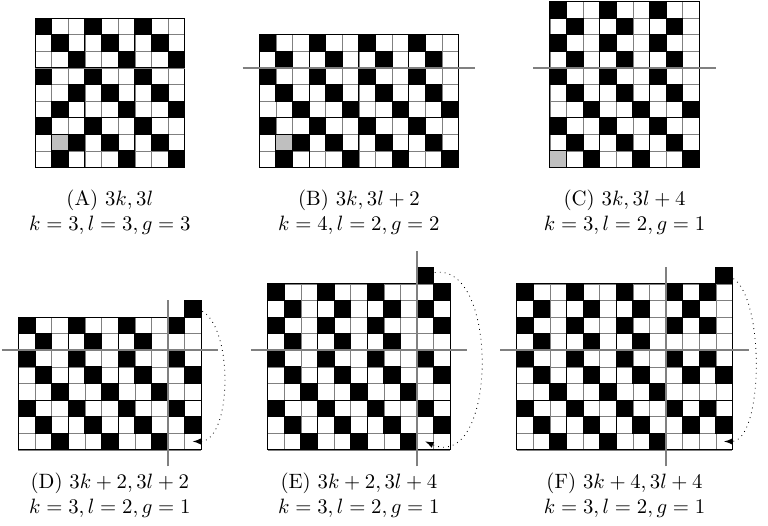} 
\end{center}
\caption{The cases for $T(m,n)$. We depict a concrete example for each of the cases. \label{fig:torc}}
\end{figure}

\begin{description}
\item[(A) {\boldmath $a = 0, b = 0$.}] We do not add anything now.

\item[(B) {\boldmath $a = 0, b = 2$.}] We tile the rectangle $3k \times 2$ with \btd.

\item[(C) {\boldmath $a = 0, b = 4$.}] We tile the rectangle $3k \times 4$ with \btc.

\item[(D) {\boldmath $a = 2, b = 2$.}] We tile the rectangle $3k \times 2$ with \btd,
  the rectangle $2 \times 3l$ with \bdt{} and place \bddw{} at $[3k,3l]$.

\item[(E) {\boldmath $a = 2, b = 4$.}] We tile the rectangle $3k \times 4$ with \btc,
  the rectangle $2 \times 3l$ with \bdt{} and place \bdcw{} at $[3k,3l]$.
  
\item[(F) {\boldmath $a = 4, b = 4$.}]
  We tile the rectangle $3k \times 4$ with \btc,
  the rectangle $4 \times 3l$ with \bct{} and place \bccw{} at $[3k,3l]$.
\end{description}

The construction is finished for cases (D), (E) and (F). Cases (A), (B) and (C)
require an extra black square. We place it at $[0,0]$ if $g=1$ and  at $[1,1]$ otherwise. It is colored
gray in Figure~\ref{fig:torc}.

Let $S$ be the set of black squares in our construction. 
In cases (A), (B) and (C) the size of $S$ is $\frac{mn}{3}+1 = \frac{mn+3}{3}$.
In cases (D) and (F) the size of $S$ is $\frac{mn+2}{3}$
and in the last case (E) the size of $S$ is $\frac{mn+4}{3}$.

Now we check the correctness of the construction; that is, we verify that the black squares constitute an irreversible 3-conversion set. We start with 
case (A) where $m = 3k$ and $n = 3l$. 

By a \emph{white cycle} we denote a connected set of white squares $W \subseteq T$ 
where every square in $W$ has at least two neighbors in $W$.
Note that the black squares of $T$ form an irreversible 
3-conversion set if and only if $T$ contains no white cycle.

\begin{observation}\label{obs:g}
Let $T(3k,3l)$ be filled with \btt. Then it contains precisely $g$ disjoint white cycles.
\end{observation}

\begin{proof}
Instead of $T(3k,3l)$, consider a rectangle $R=[0,k]\times[0,l]$ filled with unit squares, each containing a diagonal black segment. Assume without loss of generality that $k\ge l$. The rectangle $R$ together with the diagonals represent a drawing of a system of closed curves on the torus. Clearly, the number of the closed curves is equal to the number of white cycles in $T(3k,3l)$. The curves in the representation intersect $R$ in a set of disjoint diagonal segments, each intersecting the top or the bottom edge of $R$. Two diagonal segments intersecting the boundary of the rectangle at $[i,0]$ and $[i,l]$, or at $[0,i]$ and $[k,i]$, belong to the same closed curve on the torus. This further implies that two diagonal segments intersecting the bottom edge of $R$ at points $[i,0]$ and $[j,0]$ belong to the same closed curve on the torus whenever $i+l \equiv j \ (\text{mod } k)$. It is a well-known fact from number theory that the equivalence generated by such pairs $(i,j)$ has exactly $g$ classes. This implies that the diagonal segments represent exactly $g$ closed curves on the torus.
\end{proof}

\begin{figure}
\begin{center}
 \myincludegraphics{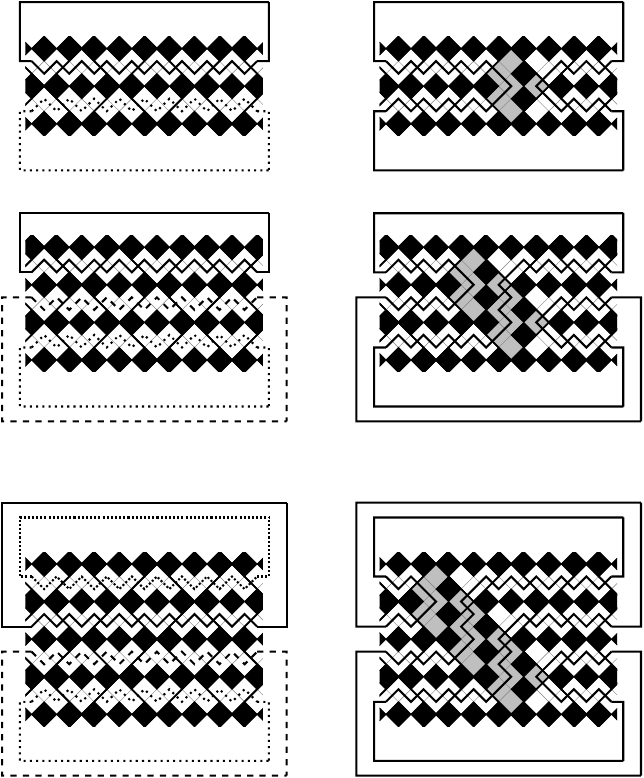} 
\end{center}
\caption{Merging two, three, or four white cycles into one. Every time we replace the pattern in one $3\times 3$ square, two adjacent white cycles are merged into one.
For better illustration, the grid is rotated 45 degrees, only a small
portion of the entire grid is depicted, and the individual cycles are depicted with distinct line styles.
The changed part is in gray color.
}
\label{fig:cykly}
\end{figure}

Start with $T(3k,3l)$ filled with \btt{} as in Observation~\ref{obs:g}.
The idea of our construction is to merge
the $g$ cycles into one long cycle by changing \btt{} to \bttx{}
in the intersection of the first column and the first $g-1$ rows; see Figure~\ref{fig:cykly}.
Finally, we add one more black square to break the resulting unique white cycle.

Observe that the small patterns used in (B)--(F) just extend the size of 
the toroidal grid but do not change the structure of white cycles from 
the $3k \times 3l$ rectangle. Thus the argument for the case  (A) can be easily 
extended to all the other cases.

This finishes the construction for the general case. 

\begin{figure}
\begin{center}
 \myincludegraphics{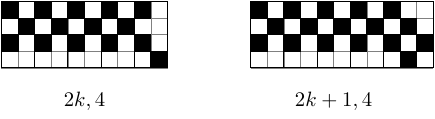} 
\end{center}
\caption{ The cases for $T(m,4)$. \label{fig:ctyri}}
\end{figure}

Now we assume without loss of generality that $n = 4$.
Let $m = 2k + a$, where $a \in \{1,2\}$.
We tile the rectangle $[0,0][2k-1,3]$ by a pattern \bdc. 
If $a = 1$ we place \bjcw{} at $[2k-1, 0]$ and if $a = 2$
we place \bdcd{} at $[2k, 0]$. The resulting grids
are depicted in Figure~\ref{fig:ctyri}.

\section*{Acknowledgements}
We thank Shuichi Ueno for his kind help with clearing our confusion regarding the parity problem for cographic matroids. We also thank J\'ozsef Balogh and Hao Huang for pointing us to the papers containing results on minimum irreversible $k$-conversion sets in high dimensional grids. Finally, we thank the anonymous reviewers for many helpful comments.

\end{document}